\documentclass[conference]{IEEEtran}

\IEEEsettopmargin{t}{0.735in}

\usepackage{cite}
\usepackage{graphicx}
\usepackage{subcaption}
\usepackage{xcolor}
\usepackage{tikz}
\usepackage{pgfplots}
\usepackage{multirow}
\usepackage{upgreek}
\usepackage{amssymb}
\usepackage{amsmath}
\usepackage{cases}
\usepackage{pifont}
\usepackage{bm}
\usepackage{amsthm}
\newtheorem{theorem}{Theorem}

\usepackage{tabularx}
\usepackage{booktabs}
\newcolumntype{Y}{>{\centering\arraybackslash}X}
\usepackage{array}
\usetikzlibrary{arrows,positioning,shapes.geometric,spy}
\usepackage{float}
\usepackage{algorithmic}
\usepackage{xcolor}
\usepackage[linesnumbered,ruled,vlined]{algorithm2e}
\usepackage{changes}
\usepackage{bbm}

\newcommand{\fixme}[2]{\ifx&#2&{\leavevmode\color{red}#1}\else{\leavevmode\color{red}FIXME\{}#1{\leavevmode\color{red}\}}\footnote{{\leavevmode\color{red}#2}}\PackageWarning{Fixme}{#1: #2}\fi}

\SetCommentSty{mycommfont}

\begin{document}

\title{On the Decoding of Polar Codes on Permuted Factor Graphs}

\author{\IEEEauthorblockN{Nghia Doan\IEEEauthorrefmark{1}, Seyyed Ali Hashemi\IEEEauthorrefmark{1}, Marco Mondelli\IEEEauthorrefmark{2}, Warren J. Gross\IEEEauthorrefmark{1}}
\IEEEauthorblockA{
\IEEEauthorrefmark{1}Department of Electrical and Computer Engineering, McGill University, Montr\'eal, Qu\'ebec, Canada\\
\IEEEauthorrefmark{2}Department of Electrical Engineering, Stanford University, California, USA\\
Email: nghia.doan@mail.mcgill.ca, seyyed.hashemi@mail.mcgill.ca, mondelli@stanford.edu, warren.gross@mcgill.ca}}

\maketitle
\begin{abstract}
Polar codes are a channel coding scheme for the next generation of wireless communications standard (5G). The belief propagation (BP) decoder allows for parallel decoding of polar codes, making it suitable for high throughput applications. However, the error-correction performance of polar codes under BP decoding is far from the requirements of 5G. It has been shown that the error-correction performance of BP can be improved if the decoding is performed on multiple permuted factor graphs of polar codes. However, a different BP decoding scheduling is required for each factor graph permutation which results in the design of a different decoder for each permutation. Moreover, the selection of the different factor graph permutations is at random, which prevents the decoder to achieve a desirable error-correction performance with a small number of permutations. In this paper, we first show that the permutations on the factor graph can be mapped into suitable permutations on the codeword positions. As a result, we can make use of a single decoder for all the permutations. In addition, we introduce a method to construct a set of predetermined permutations which can provide the correct codeword if the decoding fails on the original permutation. We show that for the 5G polar code of length $1024$, the error-correction performance of the proposed decoder is more than $0.25$~dB better than that of the BP decoder with the same number of random permutations at the frame error rate of $10^{-4}$.
\end{abstract}

\begin{IEEEkeywords}
polar codes, belief propagation decoding, permuted factor graph.
\end{IEEEkeywords}

\IEEEpeerreviewmaketitle
\section{Introduction} \label{sec:intro}
Polar codes \cite{arikan} are the first class of error-correcting codes that was proved to achieve channel capacity with efficient encoding and decoding algorithms. As such, they have been adopted for the enhanced mobile broadband (eMBB) control channel of the fifth generation (5G) wireless communications standard. Successive cancellation (SC) and belief propagation (BP) decoding algorithms are two methods introduced in \cite{arikan} to decode polar codes. Although SC decoding can provide a low-complexity solution, its serial nature prevents the decoding to reach high throughput. Furthermore, the polarization phenomenon of polar codes under SC decoding requires large block lengths. Thus, for short to moderate code lengths, the error-correction performance of SC decoding is insufficient for many practical applications. To improve the performance of SC decoding, SC list (SCL) decoding was proposed in \cite{tal_list}. SCL decoding maintains a list of most likely codewords by running a list of coupled SC decoders in parallel. An enormous improvement in the error probability of SCL decoder is obtained if a cyclic redundancy check (CRC) is used to select the correct codeword \cite{tal_list}. However, SCL decoding still suffers from the serial nature and low throughput \cite{Alexios_LLR_SCLD}, which scales with list size. Several attempts have been made to reduce the computational complexity and increase the throughput of polar code SCL decoders \cite{PSCL_GLOBECOM,Ali_mem_effic_PC,Ali_FSSCL}.

Unlike SC and SCL decoding, the message passing process of BP decoding can be executed in parallel, which allows this decoder to obtain high throughput. However, BP decoding cannot achieve low error probability with a limited number of iterations. One of the methods used to improve the error-correction performance of BP decoding for polar codes is to utilize the redundant factor graph representations \cite{Kor09thesis, hussami2009performance}. It was shown in \cite{Kor09thesis, hussami2009performance} that for a polar code of length $N$, there are $(\log_{2}N)!$ redundant representations which can be constructed by different permutations of the layers in the factor graph of polar codes. \cite{Kor09thesis, hussami2009performance} also suggested to use only the $\log_2 N$ cyclic permutations in the factor graph layers, in order to limit the large number of factor graph permutations. In \cite{elkelesh2018belief}, it was shown that if the permutations in the factor graph of polar codes are selected randomly and a CRC is concatenated to the polar code, the error probability of decoding a polar code on the permuted factor graph outperforms the performance of a non CRC-aided SCL decoder. However, the number of randomly selected permutations and therefore the number of parallel BP decoders required to achieve a reasonable error probability in \cite{elkelesh2018belief} is too high. Moreover, the BP decoder scheduling has to be changed for each permutation of the factor graph layers. Therefore, a different decoder has to be designed for each permutation in practical applications.

In this paper, we first show that the permutations on the polar code factor graph can be mapped into permutations on the codeword positions, therefore allowing the use of the same decoder structure for all the permutations. Second, based on the observation that the decoding on the original permutation only fails with a small error probability, we propose a novel empirical approach to construct a set of good permutations. Our experimental results show that for the polar code selected for the eMBB control channel of 5G, with length $1024$, rate $0.5$, and by using a $24$-bit CRC, the error-correction performance is more than $0.25$ dB better than that of the BP decoder when the same number of random permutations are employed, at the frame error rate (FER) of $10^{-4}$.

The remainder of this paper is organized as follows. Section~\ref{sec:polar} briefly introduces a background on polar codes and BP decoding. Section~\ref{sec:pcbp} and Section~\ref{sec:results} present the proposed decoder and the experimental results, respectively. Finally, concluding remarks are provided in Section~\ref{sec:conclude}.

\section{Preliminaries} \label{sec:polar}

\subsection{Polar Codes}
\label{sec:polar:polar}

A polar code $\mathcal{P}(N,K)$ of length $N$ with $K$ information bits is constructed by applying a liner transformation to the message word $\bm{u} = \{u_0,u_1,\ldots,u_{N-1}\}$ as $\bm{x} = \bm{u}\bm{G}^{\otimes n}$ where $\bm{x} = \{x_0,x_1,\ldots,x_{N-1}\}$ is the codeword, $\bm{G}^{\otimes n}$ is the $n$-th Kronecker power of the polarizing matrix $\bm{G}=\bigl[\begin{smallmatrix} 1&0\\ 1&1 \end{smallmatrix} \bigr]$, and $n = \log_2 N$. The vector $\bm{u}$ contains a set $\mathcal{A}$ of $K$ information bits and a set $\mathcal{F}$ of $N-K$ frozen bits. The positions and the value of the frozen bits are known to the encoder and the decoder. The codeword $\bm{x}$ is then modulated and sent through the channel. In this paper, binary phase-shift keying (BPSK) modulation and additive white Gaussian noise (AWGN) channel model are considered, therefore the soft vector of the transmitted codeword received by the decoder is
\begin{equation}
\bm{y}=(\mathbf{1}-2\bm{x})+\bm{z}\text{,}
\end{equation}
where $\mathbf{1}$ is an all-one vector with size $N$, and $\bm{z} \in \mathbbm{R}^N$ is the AWGN noise vector with variance $\sigma^2$ and zero mean. It is noteworthy that all the decoders presented in this paper work in the log-likelihood ratio (LLR) domain. The corresponding LLR vector of the transmitted codeword is 
\begin{equation}
\text{LLR}(\bm{x})=\ln{\frac{\text{Pr}(\bm{x}=0|\bm{y})}{\text{Pr}(\bm{x}=1|\bm{y})}}=\frac{2\bm{y}}{\sigma^2}\text{.}
\end{equation}

\subsection{Belief Propagation Decoding}
\label{sec:polar:BPD}

\begin{figure}[t]
  \centering
  \label{figBPDec}	
  \begin{subfigure}[b]{0.5\textwidth}
  	\centering
  	\includegraphics[width=0.7\linewidth]{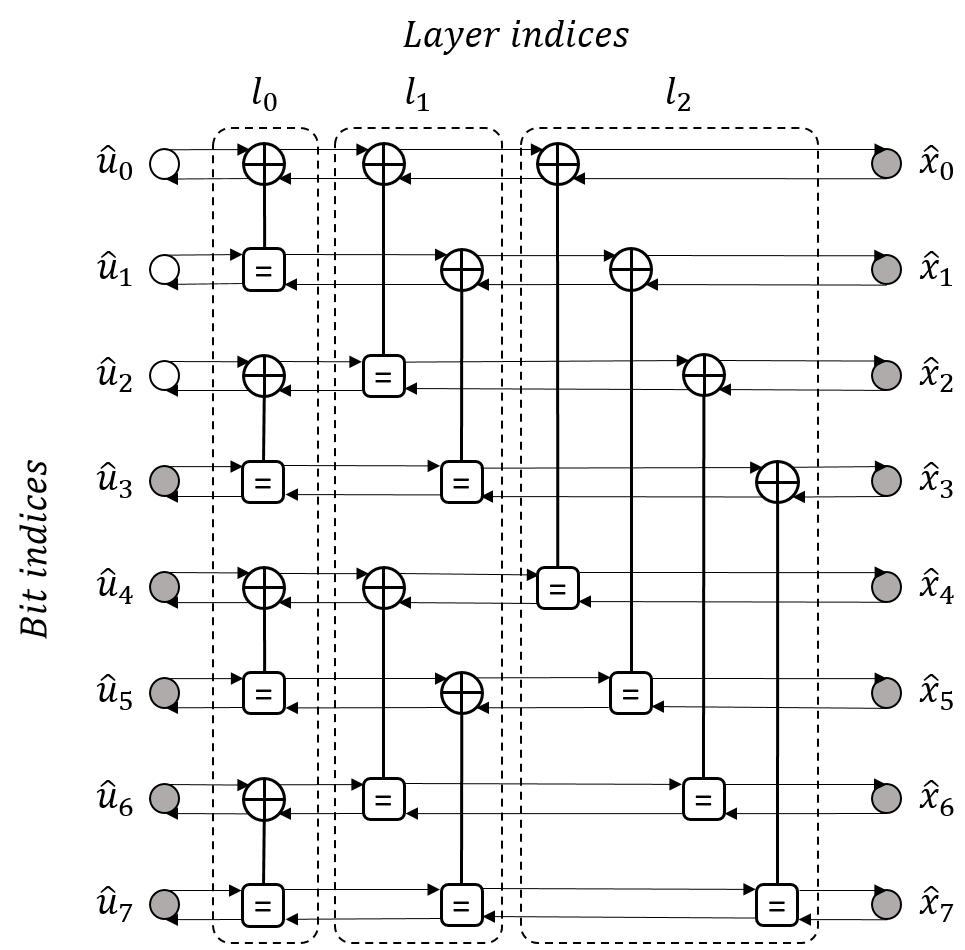}
  	\caption{}
  	\label{figBPDec:a}
  \end{subfigure}  
  \begin{subfigure}[b]{0.5\textwidth}
  	\centering
  	\includegraphics[width=0.4\linewidth]{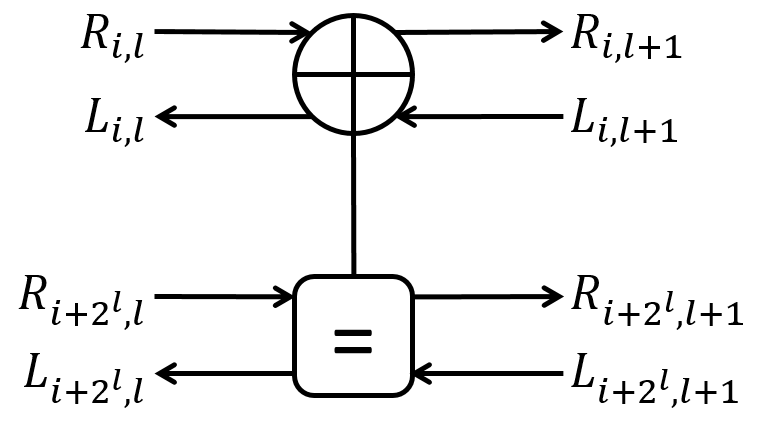}
  	\caption{}
  	\label{figBPDec:b}
  \end{subfigure}
  \caption{(a) BP decoding on the factor graph of $\mathcal{P}(8,5)$ with $\{u_0,u_1,u_2\}\in \mathcal{F}$, (b) a processing element (PE).}
\end{figure}

BP decoding is an iterative message passing algorithm applied on the factor graph representation of a code. The received channel LLR values, $\text{LLR}({\bm{x}})$, are iteratively propagated through the graph until either the LLR values converge or the maximum number of iterations is reached. The decoder then makes a hard decision based on the resulting LLR values to maximize the \textit{a posteriori} probability $\text{Pr}(\bm{x}|\bm{y})$.

Fig.~\ref{figBPDec:a} illustrates BP decoding on a factor graph representation of $\mathcal{P}(8,5)$. The messages are propagated through the processing elements (PEs) \cite{arikan2010polar} located in each layer. An update iteration starts with a right-to-left message pass that propagates the LLR values from the channel (rightmost) layer, to the information bit (leftmost) layer, and ends with the left-to-right message pass which occurs in the reverse order. Fig.~\ref{figBPDec:b} shows a PE with its corresponding messages, where $R_{i,l}$ denotes a left-to-right message, and $L_{i,l}$ denotes a right-to-left message of the $i$-th bit index at layer $l$. The update rule \cite{arikan2010polar} of the messages in each PE is given as
\begin{align}
\begin{split}
L_{i,l} &= f(L_{i,l+1},R_{i+2^l,l}+L_{i+2^l,l+1}) \text{,}\\  
L_{i+2^l,l} &= f(L_{i,l+1},R_{i,l}) + L_{i+2^l,l+1} \text{,}\\ 
R_{i,l+1} &= f(R_{i,l},L_{i+2^l,l+1}+R_{i+2^l,l})  \text{,}\\
R_{i+2^l,l+1} &= f(R_{i,l},L_{i,l+1}) + R_{i+2^l,l} \text{,}
\end{split}
\label{PEUpdate}
\end{align}
where $f(x,y)=\ln{\frac{1+xy}{x+y}}$ for any $x,y \in \mathbbm{R}$.

Initially, $L_{i,n-1}$ ($0 \leq i < N$) are set to the received channel LLR values, $R_{i,0}$ ($i \in \mathcal{F}$) are set to $+\infty$, and all the other left-to-right and right-to-left messages of the PEs are set to $0$. The BP decoding performs a predetermined $I_{\max}$ update iterations where the messages are propagated through all PEs in accordance with (\ref{PEUpdate}). The decoder then makes a hard decision on the LLR values of the $i$-th bit at the information bit layer to obtain the estimated message word as
\begin{equation}
\hat{u}_i=
  \begin{cases}
    0 \text{,} & \text{if } R_{i,0}+L_{i,0} \geq 0 \text{,}\\
    1 \text{,} & \text{otherwise.}
  \end{cases} 
\end{equation}

In order to reduce the decoding latency, $G$-matrix-based \cite{yuan2014early} or CRC-based \cite{elkelesh2018belief} early termination conditions are checked after each update iteration. If the early termination condition is satisfied, the decoding is terminated with the assumption that the LLR values have converged to indicate the correct codeword.

\subsection{BP Decoding on Permuted Factor Graphs}
\label{sec:polar:PGBPD}

\begin{figure*}[t]
\centering
\includegraphics[width=0.65\linewidth]{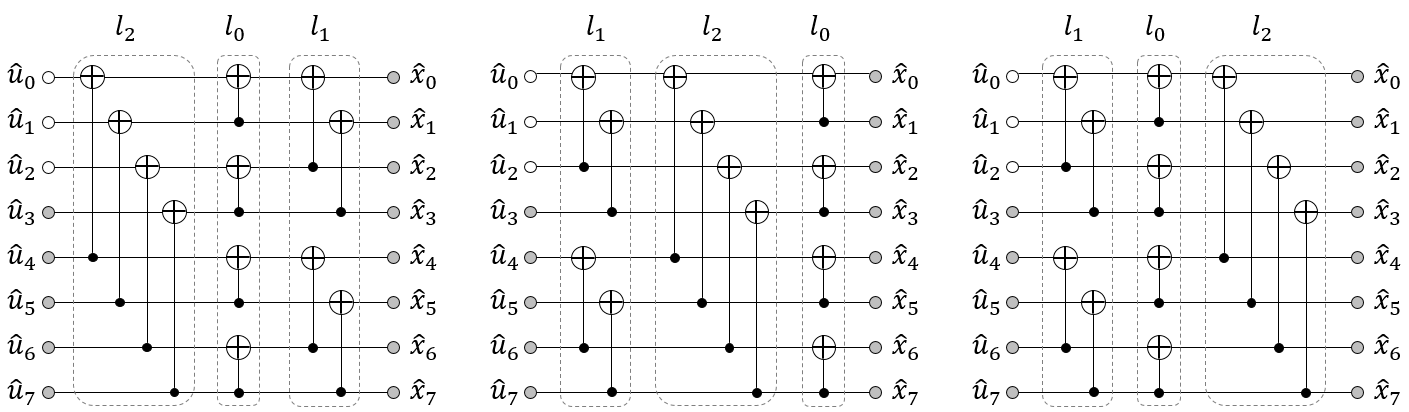}
\caption{Permuted factor graph representations for $\mathcal{P}(8,5)$.}
\label{PerTG}
\end{figure*}

\begin{figure*}[t]
\centering
\includegraphics[width=0.65\linewidth]{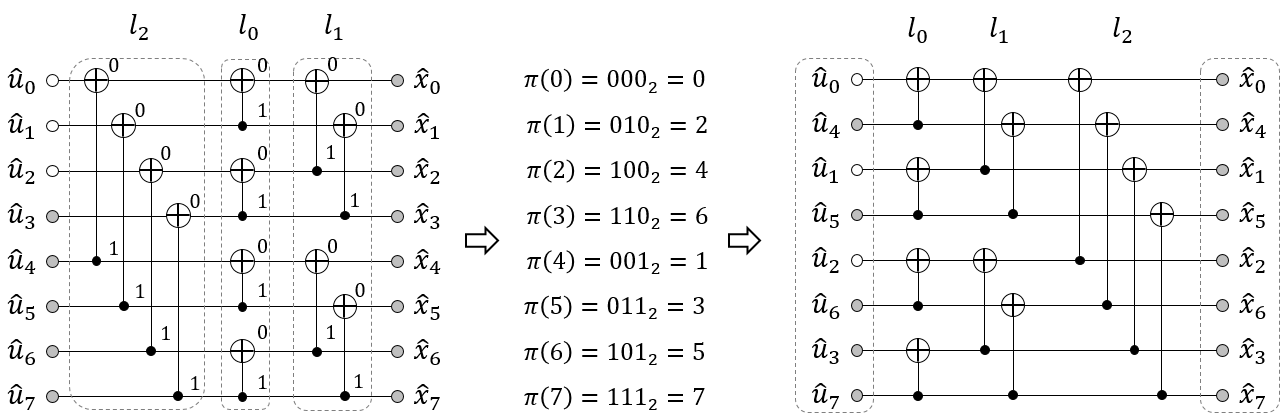}
\caption{The proposed mapping from factor graph permutation to bit-index permutation for $\mathcal{P}(8,5)$.}
\label{BitRelization}
\end{figure*}

Factor graph permutations are a way to provide multiple representations of a single code. It was observed in \cite{Kor09thesis, hussami2009performance} that there exists $n!$ different ways to represent a polar code by permuting the layers in its factor graph. Fig.~\ref{PerTG} illustrates such permutations for $\mathcal{P}(8,5)$, where $3$ out of $3!=6$ permutations are shown. Note that the two leftmost factor graphs in Fig.~\ref{PerTG} are formed by applying cyclic shifts to the original factor graph depicted in Fig.~\ref{figBPDec:a}. In \cite{elkelesh2018belief}, parallel BP decoders are applied on a set of randomly selected factor graphs of a polar code concatenated with a CRC. Although this decoding scheme shows improvement in error probability when compared to a non CRC-aided SCL decoder, permuting layers results in different BP scheduling which consequently requires the design of a different BP decoder for each permutation. Moreover, a large number of random permutations are required to achieve a reasonable error-correction performance, which makes this decoding scheme too complex for practical applications.

\section{Improved BP Decoding on Permuted Factor Graphs} \label{sec:pcbp}

In this section, we first show that there is a one-to-one mapping between the permutations on the factor graph layers and the permutations on the codeword positions. We then propose a method to select the factor graph permutations that improves the error-correction performance of polar codes when decoding is performed on the permuted factor graphs. 

\subsection{From Permutations on Factor Graph Layers to Permutations on Codeword Positions}

We denote by $\{l_{n-1},\ldots,l_0\}$ the layers of the original factor graph of polar codes (the one represented in the right part of Fig. \ref{BitRelization}) and by $\{b^{(i)}_{n-1},\ldots,b^{(i)}_0\}$ the binary expansion of the integer $i$.  

\begin{theorem} \label{th:perm}
Let $\pi : \{0, \ldots, n-1\} \to \{0, \ldots, n-1\}$ be a permutation. Then, the synthetic channel associated to the position with binary expansion $\{b^{(i)}_{n-1},\ldots,b^{(i)}_0\}$ on the factor graph with layers $\{l_{n-1},\ldots,l_{0}\}$ is the same as the synthetic channel associated to the position with binary expansion $\{b^{(i)}_{\pi(n-1)},\ldots,b^{(i)}_{\pi(0)}\}$ on the factor graph with layers $\{l_{\pi(n-1)},\ldots,l_{\pi(0)}\}$.
\end{theorem}

\begin{proof}
Consider the original factor graph of polar codes. Then, the synthetic channel associated to the position with binary expansion $\{b^{(i)}_{n-1},\ldots,b^{(i)}_0\}$ is given by
$$
(((W^{(b^{(i)}_{n-1})})^{(b^{(i)}_{n-2})})^{\cdots})^{(b^{(i)}_0)},
$$
where $W$ is the transmission channel and the transformations $W\to W^{(0)}$ and $W\to W^{(1)}$ are the ``minus'' and ``plus'' polar transforms formally defined in (2.3) and (2.4) of \cite{Kor09thesis}. Note that the $i$-th component of the message word $u_i$ is connected to the $i$-th component of the codeword $x_i$ by associating a ``XOR'' to a $0$ and a ``dot'' to a $1$ in the binary expansion of the integer $i$. Note also that, by permuting the layers of the factor graph, we simply permute the order of those ``XOR''s and ``dot''s operations. Hence, the effect of applying a permutation $\pi$ to the binary expansion of $i$ is the same as the effect of applying the same permutation on the layers of the factor graph. 
\end{proof}

Fig.~\ref{BitRelization} shows an example of the proposed mapping applied to a permuted factor graph of $\mathcal{P}(8,5)$. It can be seen that, by using the permuted bit indices, the structure of the factor graph is unchanged. Therefore, the original decoder can be used to perform decoding on all the required permutations. In addition, the proposed mapping allows to use other decoding algorithms, such as SC and SCL, on the permuted factor graphs without changing the decoder structure. This is particularly useful for hardware implementation.

\subsection{Selection of Good Permutations}
\label{sec:proposed:b}

First, we observe that BP decoding only fails to decode with a small probability on the original factor graph of polar codes. Therefore, we keep the original factor graph in the set of good permutations and create a set of permutations that provides the correct codeword when the decoder fails on the original factor graph. To this end, we first provide an approach to construct the permutation set $\mathbbm{P}$ of a polar code with $n$ factor graph layers which represents a search space for the good factor graph permutations. This approach is summarized in Algorithm~\ref{alg1}, which uses the recursive function \texttt{FormPermutation} described in Algorithm~\ref{alg2}. The main idea is to recursively create $\mathbbm{P}$ by permuting only some of the layers on the right-most side of the graph (see Fig.~\ref{figBPDec:a}), in order to limit the size of the search space. In other words, we fix the layers $\{l_0, \ldots, l_{n-k-1}\}$, with $0 \leq k < n$, and we consider permutations only of the remaining layers $\{l_{n-k},\ldots,l_{n-1}\}$. Note that the size of the search space $\mathbbm{P}$ is $|\mathbbm{P}| = k!$.

\begin{algorithm}[t]
\DontPrintSemicolon
\caption{Forming the permutation set }
\label{alg1}
\SetKwInOut{Input}{Input}
\SetKwInOut{Output}{Output}
\SetKwInOut{Initialization}{Initialization}
\SetKwFunction{FormPermutation}{FormPermutation}
\SetKwFunction{Main}{Main}

\Input{$n,k$} 
\Output{$\mathbbm{P}$} 
\tcp{Fixed layers}
$p_{fix} \leftarrow \{l_0,l_1,\ldots,l_{n-k-1}\}$ \\
\tcp{Permuted layers}
$p_{per} \leftarrow \{l_{n-k},l_{n-k+1},\ldots,l_{n-1}\}$\\ 
\tcc{Recursively form the permutation set}
$\mathbbm{P}_{per} \leftarrow$ \FormPermutation{$p_{per}$}\\
\tcc{Concatenate $p_{fix}$ to each element of $\mathbbm{P}_{per}$}
$\mathbbm{P} \leftarrow \emptyset$\\
\For{$j \leftarrow 0$ \KwTo $|\mathbbm{P}_{per}|-1$}{
	$\mathbbm{P} \leftarrow \mathbbm{P} \cup [p_{fix},\mathbbm{P}_{per}[j]]$
}
\Return{$\mathbbm{P}$}
\end{algorithm}

\begin{algorithm}[t]
\DontPrintSemicolon
\caption{\texttt{FormPermutation($p_{per}$)}}
\label{alg2}
\SetKwFunction{FormPermutation}{FormPermutation}
\SetKwInOut{Input}{Input}
\SetKwInOut{Output}{Output}
\Input{$p_{per}$ \tcc*{Factor graph layer order}}  
\Output{$\mathbbm{P}_{per}$ \tcc*{Set of all permutations}} 
$\mathbbm{P}_{per} \leftarrow \emptyset$ \\
\eIf{$length(p_{per})=2$}{
		$\mathbbm{P}_{per}[0] \leftarrow [p_{per}[0],p_{per}[1]]$\\
		$\mathbbm{P}_{per}[1] \leftarrow [p_{per}[1],p_{per}[0]]$
}{
		\For{$i\leftarrow 0$ \KwTo $length(p_{per})-1$}{
		$p^{'}_{per} \leftarrow p_{per}$ \\
		$l_{i} \leftarrow p^{'}_{per}[i]$  \\
		$p^{'}_{per} \leftarrow p^{'}_{per}\setminus l_{i}$\\
	      $\mathbbm{P}^{'}_{per} \leftarrow$ \FormPermutation{$p^{'}_{per}$} \\		
		\For{$j \leftarrow 0$ \KwTo $|\mathbbm{P}^{'}_{per}|-1$}{
		$\mathbbm{P}_{per} \leftarrow \mathbbm{P}_{per} \cup [l_{i},\mathbbm{P}^{'}_{per}[j]]$\\
		}
	}
}
\Return{$\mathbbm{P}_{per}$}
\end{algorithm}

At this point, we evaluate numerically the probability of correct decoding for each element in $\mathbbm{P}$, when the decoding fails on the original factor graph of polar codes. Then, we select the $M$ elements in $\mathbbm{P}$ with the highest probability of successful decoding and form the subset $\mathbbm{P}_b$ of $\mathbbm{P}$. Note that $|\mathbbm{P}_b| = M$. At the decoder, we consider only the $M$ permutations in $\mathbbm{P}_b$. This is described in Algorithm~\ref{alg3}. Note that Theorem~\ref{th:perm} is used to perform decoding on the permuted bit indices instead of the permuted factor graph layers. In addition, all the BP decoders run a maximum of $I_{\max}$ iterations and an early stopping criteria is used to terminate the BP decoding process if one of the BP decoders satisfies the termination condition. If none of the BP decoders satisfy the early stopping criteria or the maximum number of iterations is reached, the message word given by the original permutation is selected as the final decoding result.

\begin{algorithm}[t]
\DontPrintSemicolon
\caption{BP decoding on predetermined bit-index permutations}
\label{alg3}
\SetKwInOut{Input}{Input}
\SetKwInOut{Output}{Output}
\SetKwFunction{BPDecoding}{BPDecoding}
\Input{$\bm{y},\mathbbm{P}_b$}
\Output{$\bm{\hat{u}}$}
\tcc{Initilize early termination flag}
$isTerm \leftarrow False$\\
\For{all permutations in $\mathbbm{P}_b$}{	
	\tcc{Permute $\bm{y}$ using Theorem~\ref{th:perm}}
	\For{$i \leftarrow 0$ \KwTo $N-1$}{	
	$\bm{y_{\pi}}[\pi(i)] \leftarrow \bm{y}[i]$ \\
	}
	\tcc{Apply BP decoding on $\bm{y_{\pi}}$}
	$[\bm{\hat{u}_{\pi}},isTerm] \leftarrow$ \BPDecoding{$\bm{y_{\pi}}$}\\
	\tcc{Obtain $\hat{u}$ if early termination}
	\If{$(isTerm=True)$}{
		\For{$i \leftarrow 0$ \KwTo $N-1$}{
		$\bm{\hat{u}}[i] \leftarrow \bm{\hat{u}_\pi}[\pi(i)]$\\
		}
		\Return{$\bm{\hat{u}}$}
	}
}
\tcc{If no early termination}
\If{$(isTerm=False)$}{
	$\bm{\hat{u}} \leftarrow$ \BPDecoding{$\bm{y}$}\\
	\Return{$\bm{\hat{u}}$}
}
\end{algorithm}

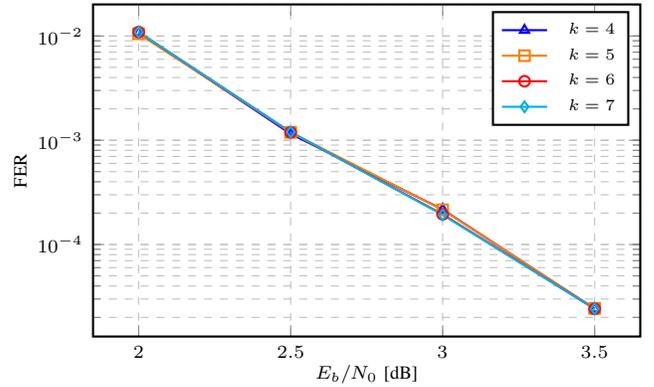
\begin{figure}[t]
\centering
\begin{tikzpicture}[spy using outlines = {rectangle, magnification=2, connect spies}]
  \pgfplotsset{	
    label style = {font=\fontsize{7pt}{7}\selectfont},
    tick label style = {font=\fontsize{7pt}{7}\selectfont}
  }

\begin{axis}[
	scale = 1,
    ymode=log,
    xlabel={$E_b/N_0$ [\text{dB}]}, xlabel style={yshift=0.8em},
    ylabel={FER}, ylabel style={yshift=-0.75em},
	xtick={2,2.5,3,3.5},
    grid=both,
    ymajorgrids=true,
    xmajorgrids=true,
    grid style=dashed,
    width=1\columnwidth, height=6cm,
    thick,
    mark size=2,
    legend style={
      column sep= 2mm,
      font=\fontsize{6pt}{7.2}\selectfont,
    },
]

\addplot[
    color=blue,
    mark=triangle,
    thick,
    mark size=2,
]
table {
2	0.0108	
2.5	0.0011526	
3	0.000216648	
3.5	2.41575e-05	
};
\addlegendentry{$k=4$}

\addplot[
    color=orange,
    mark=square,
    thick,
    mark size=2,
]
table {
2	0.0105	
2.5	0.00119033	
3	0.000215407	
3.5	2.43474e-05	
};
\addlegendentry{$k=5$}

\addplot[
    color=red,
    mark=o,
    thick,
    mark size=2,
]
table {
2	0.0109	
2.5	0.00119033	
3	0.000194821
3.5	2.41575e-05	
};
\addlegendentry{$k=6$}

\addplot[
    color=cyan,
    mark=diamond,
    thick,
    mark size=2,
]
table {
2	0.0109	
2.5	0.00119959	
3	0.000192324	
3.5	2.41439e-05	
};
\addlegendentry{$k=7$}

\end{axis}
\end{tikzpicture} 
\centering
\caption{FER performance of BP decoding on $16$ best permuted factor graphs of $\mathcal{P}(1024,512)$ with a $24$-bit CRC for various sizes of $\mathbbm{P}$, where $|\mathbbm{P}|=k!$.}
\label{perf:k}
\end{figure}

It should be noted that the computational complexity of forming $\mathbbm{P}_b$ is directly proportional to the size of $\mathbbm{P}$, since all the elements in $\mathbbm{P}$ are evaluated numerically. However, this operation is performed only once and is done off-line so there is no complexity overhead in the decoding process. In addition, the error-correction performance of BP decoding on permuted factor graphs selected from $\mathbbm{P}_b$ does not change significantly even for small values of $k$. Fig.~\ref{perf:k} illustrates the FER curves of applying BP decoding on permuted factor graphs in $\mathbbm{P}_b$ to decode $\mathcal{P}(1024,512)$ with a $24$-bit CRC. We set $M=16$ and we consider different values of $k$ ($k \in \{4,5,6,7\}$) corresponding to different sizes of $\mathbbm{P}$ ($|\mathbbm{P}| \in \{24,120,720,5040\}$). It can be seen that the FER performance of BP decoding on permuted factor graphs for all cases is similar, showing that the good permutations are those obtained by permuting only the layers on the right-most side of the polar codes factor graph. Therefore, the size of the search space $\mathbbm{P}$ can be significantly reduced, in order to form $\mathbbm{P}_b$ with low computational complexity.

\section{Experimental Results} \label{sec:results}

In this section, we evaluate applying the proposed technique on the decoding of the polar code $\mathcal{P}(1024,512)$, which is selected for the 5G eMBB control channel \cite{3gpp_report}, in terms of FER performance and average decoding latency. All the BP decoders considered in this section has the same number of $I_{\max} = 200$ maximum iterations.

\begin{figure}[t]
\centering
  \begin{tikzpicture}[spy using outlines = {rectangle, magnification=2, connect spies}]
  \pgfplotsset{	
    label style = {font=\fontsize{7pt}{7}\selectfont},
    tick label style = {font=\fontsize{7pt}{7}\selectfont}
  }

\begin{axis}[
	scale = 1,
    ymode=log,
    xlabel={$E_b/N_0$ [\text{dB}]}, xlabel style={yshift=0.8em},
    ylabel={FER}, ylabel style={yshift=-0.75em},
	xtick={2,2.5,3,3.5},
    grid=both,
    ymajorgrids=true,
    xmajorgrids=true,
    grid style=dashed,
    width=1\columnwidth, height=8cm,
    thick,
    mark size=2,
    legend style={
      anchor={center},
      cells={anchor=west},
      column sep= 2mm,
      font=\fontsize{6pt}{7.2}\selectfont,
    },
    legend to name=perf-legend-crc0,
    legend columns=4,
]

\addplot[
    color=black,    
    mark=square,
    thick,
    mark size=1.5,
]
table {
2	0.0216	
2.5	0.0040093	
3	0.0011595	
3.5	0.000216666	
};
\addlegendentry{PBP-CS \cite{hussami2009performance}}

\addplot[
    color=blue,
    mark=x,
    thick,
    mark size=2,
]
table {
2	0.0221
2.5	0.00405088
3	0.00112514
3.5	1.72E-04
};
\addlegendentry{PBP-R$10$ \cite{elkelesh2018belief}}

\addplot[
    color=orange,
    mark=square,
    thick,
    mark size=2,
]
table {
2	1.36E-02
2.5	2.51E-03
3	0.000710157
3.5	9.26E-05
};
\addlegendentry{PBP-B$10$}

\addplot[
    color=cyan,
    every mark/.append style={solid, fill=cyan},
    mark=diamond,
    thick,
    mark size=2,
]
table {
2	0.0319
2.5	0.00279111
3	3.19E-04
3.5	4.03E-05
};
\addlegendentry{PSC-B$10$}

\addplot[
    color=red,
    mark=*,
    thick,
    mark size=1.5,
]
table {
2	0.0978
2.5	0.0148
3	0.00176947
3.5	2.39E-04
};
\addlegendentry{SC}

\addplot[
    color=olive,
    densely dashed,
    every mark/.append style={solid, fill=olive},
    mark=*,
    thick,
    mark size=1.5,
]
table {
2	0.029
2.5	0.0069
3	0.00161791
3.5	0.000340599
};
\addlegendentry{BP}

\addplot[
    color=black,
    densely dotted,
    every mark/.append style={solid, fill=black},
    mark=*,
    thick,
    mark size=1.5,
]
table {
2	0.0084
2.5	0.00140308
3	0.000270167
3.5	3.35E-05
};
\addlegendentry{SCL$32$}

\end{axis}
\end{tikzpicture} 
  \centering
  \hspace{10pt}\ref{perf-legend-crc0}\vspace{2pt}
  \caption{FER performance of decoding $\mathcal{P}(1024,512)$ on permuted factor graphs considering $M=10$ best permutations, in comparison with considering $10$ random and cyclic shift permutations, when no CRC is used.}
  \label{perf:a}
\end{figure}
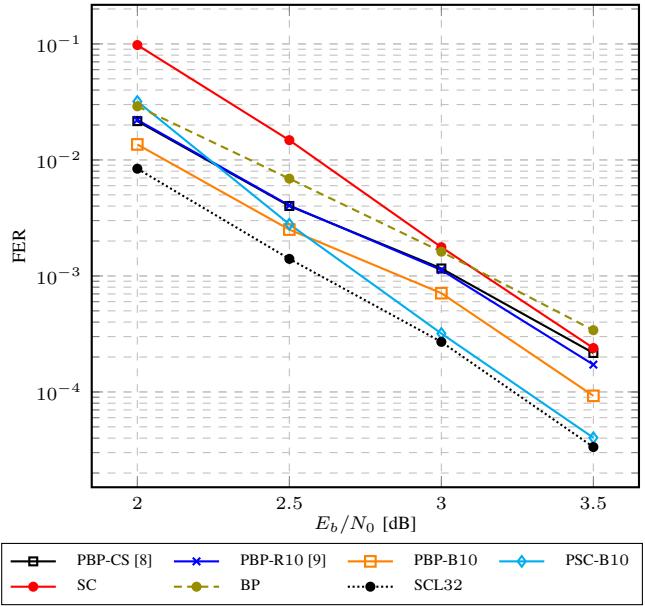

Fig.~\ref{perf:a} shows the FER performance of decoding $\mathcal{P}(1024,512)$ on permuted factor graphs (permuted bit indices) when no CRC is used. Since there are $10$ cyclic shift permutations for $\mathcal{P}(1024,512)$, we used $M = 10$ in order to have a fair comparison with other decoders. In Fig.~\ref{perf:a}, PBP-CS denotes BP decoding on permuted factor graphs with $n$ cyclic shift permutations \cite{Kor09thesis, hussami2009performance}, PBP-R$M$ denotes BP decoding on permuted factor graphs with $M$ random permutations \cite{elkelesh2018belief}, and PBP-B$M$ denotes the proposed BP decoding on permuted factor graphs with $M$ predetermined permutations in $\mathbbm{P}_b$. The results are also compared with SC, BP, and SCL decoding with list size $32$ (SCL$32$). It can be seen that while the FER performance of PBP-CS and PBP-R$10$ are similar, our proposed decoder results in $0.15$~dB improvement in comparison with both PBP-CS and PBP-R$10$ at FER$=10^{-3}$. However, there is still a $0.25$~dB gap between the FER performance of the proposed BP decoding on permuted factor graphs and that of SCL$32$. Using the result of Theorem~\ref{th:perm}, we performed SC decoding on permuted factor graphs and it can be seen in Fig.~\ref{perf:a} that the FER performance gap between using $M=10$ best permutations for SC decoding (denoted by PSC-B$M$) and that of SCL$32$ is less than $0.05$~dB in the high $E_b/N_0$ region. It should be noted that in order to select the correct codeword among the $M$ SC decoders, we used the absolute LLR value of the last decoded bit as a reliability measure of the decoding process, which is also used in \cite{Condo_BD}.

\begin{figure}[t]
  \centering
  \begin{tikzpicture}[spy using outlines = {rectangle, magnification=1.8, width=1cm, height=3.5cm, connect spies}]
  \pgfplotsset{	
    label style = {font=\fontsize{7pt}{7}\selectfont},
    tick label style = {font=\fontsize{7pt}{7}\selectfont}
  }

\begin{axis}[
	scale = 1,
    ymode=log,
    xlabel={$E_b/N_0$ [\text{dB}]}, xlabel style={yshift=0.8em},
    ylabel={FER}, ylabel style={yshift=-0.75em},
	xtick={2,2.5,3,3.5},
    grid=both,
    ymajorgrids=true,
    xmajorgrids=true,
    grid style=dashed,
    width=1\columnwidth, height=9cm,
    thick,
    mark size=2,
    legend style={
      anchor={center},
      cells={anchor=west},
      column sep= 1mm,
      font=\fontsize{6pt}{7.2}\selectfont,
    },
    legend to name=perf-legend-crc24,
    legend columns=3,
]

\addplot[
    color=black,
    every mark/.append style={solid, fill=black},
    mark=*,
    thick,
    mark size=1.5,
]
table {
2	0.0084
2.5	0.00140308
3	0.000270167
3.5	3.35E-05
};
\addlegendentry{SCL$32$}

\addplot[
    color=darkgray,
    densely dashed,
    every mark/.append style={solid, fill=darkgray},
    mark=*,
    thick,
    mark size=1.5,
]
table {
2	0.00074985
2.5	8.05E-06
};
\addlegendentry{SCL$32$-CRC$24$}

\addplot[
    color=black,    
    mark=square,
    thick,
    mark size=1.5,
]
table {
2	0.0391	
2.5	0.00360959	
3	0.000678573	
3.5	7.18741e-05	
};
\addlegendentry{PBP-CS-CRC$24$ \cite{hussami2009performance}}

\addplot[
    color=blue,
    densely dashdotted,
    every mark/.append style={solid, fill=blue},
    mark=x,
    thick,
    mark size=2,
]
table {
2	0.0423
2.5	0.0052
3	0.000794079
3.5	9.62E-05
};
\addlegendentry{PBP-R$10$-CRC$24$ \cite{elkelesh2018belief}}

\addplot[
    color=orange,
    mark=x,
    thick,
    mark size=2,
]
table {
2	0.014	
2.5	0.0017651	
3	0.000288912	
3.5	4.06457e-05	
};
\addlegendentry{PBP-B$10$-CRC$24$}

\addplot[
    color=cyan,
    mark=star,
    thick,
    mark size=2,
]
table {
2	0.1079
2.5	0.0081
3	0.000787538
3.5	0.000100995
};
\addlegendentry{PSC-B$10$-CRC$24$}

\addplot[
    color=blue,
    densely dashdotted,
    every mark/.append style={solid, fill=blue},
    mark=o,
    thick,
    mark size=2,
]
table {
2	0.0333
2.5	0.00291886
3	0.000242761
3.5	2.53E-05
};
\addlegendentry{PBP-R$32$-CRC$24$ \cite{elkelesh2018belief}}

\addplot[
    color=orange,
    mark=o,
    thick,
    mark size=2,
]
table {
2	0.0066
2.5	0.000670988
3	6.50E-05
3.5	1.13E-05
};
\addlegendentry{PBP-B$32$-CRC$24$}

\addplot[
    color=cyan,  
    mark=otimes,
    thick,
    mark size=2,
]
table {
2	0.0817
2.5	0.0054
3	0.000265109
3.5	3.21E-05
};
\addlegendentry{PSC-B$32$-CRC$24$}

\addplot[
    color=blue,
    densely dashdotted,
    every mark/.append style={solid, fill=blue},
    mark=diamond,
    thick,
    mark size=2,
]
table {
2	0.0228
2.5	1.31E-03
3	6.31E-05
3.5	3.67E-06
};
\addlegendentry{PBP-R$128$-CRC$24$ \cite{elkelesh2018belief}}

\addplot[
    color=orange,
    mark=diamond,
    thick,
    mark size=2,
]
table {
2	0.00243831
2.5	0.000180901
3	1.46E-05
3.5	2.08E-06
};
\addlegendentry{PBP-B$128$-CRC$24$}

\addplot[
    color=cyan,   
    mark=triangle,
    thick,
    mark size=2,
]
table {
2	0.0522
2.5	0.00268586
3	0.000109915
3.5	9.01E-06
};
\addlegendentry{PSC-B$128$-CRC$24$}

\end{axis}
\end{tikzpicture} 
  \hspace{10pt}\ref{perf-legend-crc24}\vspace{2pt}
  \caption{FER performance of decoding $\mathcal{P}(1024,512)$ on permuted factor graphs considering $M \in \{10,32,128\}$ best permutations, in comparison with considering the same number of random permutations, when a $24$-bit CRC is used.}
  \label{perf:b}
\end{figure}
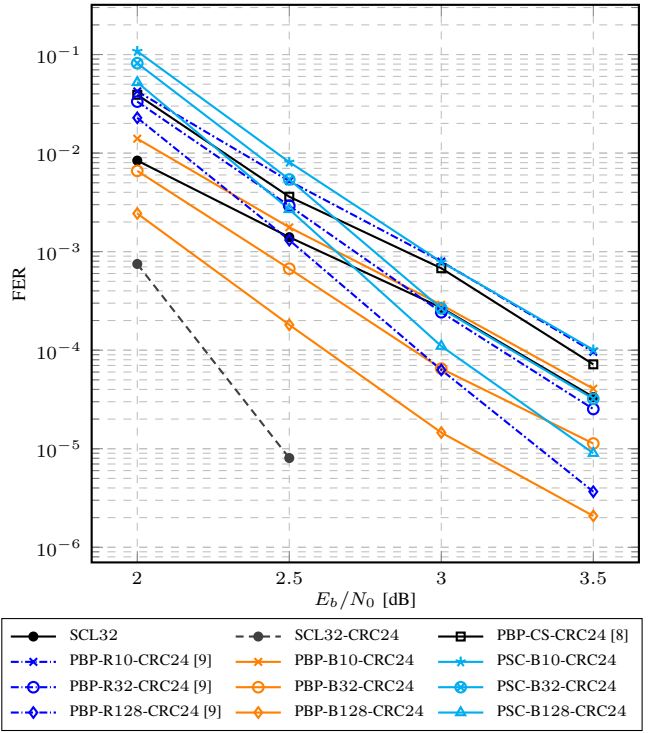

Fig.~\ref{perf:b} presents the FER performance of decoding $\mathcal{P}(1024,512)$ when a CRC of length $24$ is used to aid the decoding process. We use the CRC which is selected to be used in 5G together with polar codes with the CRC polynomial
\begin{align}
\text{CRC}24 = &x^{24} + x^{23} + x^{21} + x^{20} + x^{17} + x^{15} + x^{13} + x^{12} \nonumber \\
               &+ x^{8} + x^{4} + x^{2} + x + 1 \text{.}
\end{align}
The curves in Fig.~\ref{perf:b} are obtained by considering $M \in \{10,32,128\}$ and the suffix -CRC$24$ is used to denote that the decoding is performed with the CRC. It can be seen that when $M=10$, BP decoding with cyclic shift permutations results in a better error-correction performance than when the permutations are selected randomly. Our proposed method results in an additional $0.2$~dB performance improvement over cyclic shift permutations when $M=10$, at FER$=10^{-4}$. When $M=32$ and $M=128$, BP decoding with the proposed selection of permutations results in around $0.3$~dB improvement in comparison with BP decoding on randomly selected factor graphs. We have also performed SC decoding on the permuted factor graphs by using CRC and it can be seen in Fig.~\ref{perf:b} that when adding a CRC, the FER performance of the BP decoder is superior to that of the SC decoder. This is in contrast to the results in Fig.~\ref{perf:a} where no CRC is used.

Although the FER performance of the BP decoder on permuted factor graphs can be improved significantly by using the proposed method to select the best permutations, it is still far from that of SCL$32$-CRC$24$ for practical values of $M$ as depicted in Fig.~\ref{perf:b}. However, the parallel nature of BP decoding makes this decoder interesting for applications with stringent latency requirements. It should be noted that we can evaluate the latency of a decoder by the number of time steps required to finish the decoding process \cite{arikan}. For BP decoding, the average decoding latency can be given as
\begin{equation}
\mathcal{T}_{\text{BP}}=2nI_{\text{avg}}\text{,} \label{T_BP}
\end{equation}
where $I_{\text{avg}}$ is the average number of iterations required to finish the decoding process. We now evaluate the average decoding latency of BP decoding on permuted factor graphs of $\mathcal{P}(1024,512)$ as shown in Fig.~\ref{latency}, when the permutations are selected based on the method in this paper, in comparison with cyclic shift \cite{Kor09thesis, hussami2009performance} and random \cite{elkelesh2018belief} permutations. Note that CRC$24$ is used to early terminate the decoding when one of the decoders passes the CRC test. It can be seen that the average decoding latency of the proposed method $\mathcal{T}_{\text{PBP-B$M$-CRC$24$}}$, is less than $330$ time steps for $2 \leq E_b/N_0 \leq 3.5$, and it is always less than the average decoding latency requirements of \cite{Kor09thesis, hussami2009performance} ($\mathcal{T}_{\text{PBP-CS-CRC$24$}}$) and \cite{elkelesh2018belief} ($\mathcal{T}_{\text{PBP-R$M$-CRC$24$}}$) when $M\in\{10,32\}$. This latency saving is more significant in the low $E_b/N_0$ region. At $E_b/N_0=3$~dB, the average decoding latency of the proposed decoder with $M=32$ and CRC$24$ in accordance with (\ref{T_BP}) is around $100$ time steps. In comparison, the SCL$32$ decoder in \cite{Alexios_LLR_SCLD} requires $2582$ time steps to decode $\mathcal{P}(1024,512)$ with CRC$24$, and the Fast-SSCL-SPC decoder in \cite{Ali_FSSCL} requires $673$ time steps to decode the same code.

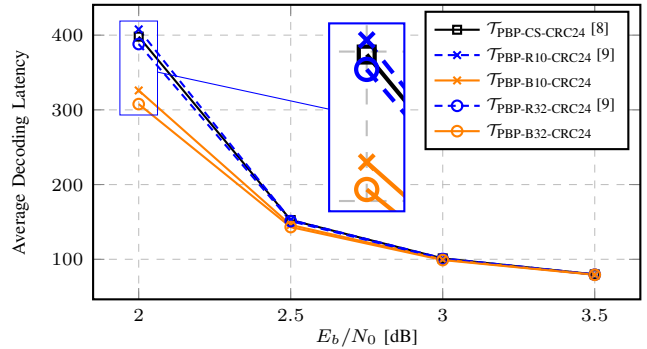
\begin{figure}[t]
  \centering
  \begin{tikzpicture}[spy using outlines = {rectangle, magnification=2, width=1cm, height=2.5cm, connect spies}]
  \pgfplotsset{	
    label style = {font=\fontsize{7pt}{7}\selectfont},
    tick label style = {font=\fontsize{7pt}{7}\selectfont}
  }

\begin{axis}[
    scale = 1,
    xlabel={$E_b/N_0$ [\text{dB}]}, xlabel style={yshift=0.8em},
    ylabel={Average Decoding Latency}, ylabel style={yshift=-0.75em},
    xtick={2,2.5,3,3.5},
    grid=both,
    ymajorgrids=true,
    xmajorgrids=true,
    grid style=dashed,
    legend style={
      cells={anchor=west},
      font=\fontsize{6pt}{7.2}\selectfont,
    },
    width=1\columnwidth, height=5.5cm,
    thick,
    mark size=3,
]

\addplot[
    color=black,    
    mark=square,
    thick,
    mark size=1.5,
]
table {
2	398.022
2.5	152.5268
3	101.1956
3.5	79.7926
};
\addlegendentry{$\mathcal{T}_{\text{PBP-CS-CRC24}}$ \cite{hussami2009performance}}

\addplot[
    color=blue,
    densely dashed,
    every mark/.append style={solid, fill=blue},
    mark=x,
    thick,
    mark size=2,
]
table {
2	407.73
2.5	152.7248
3	101.3612
3.5	79.7926
};
\addlegendentry{$\mathcal{T}_{\text{PBP-R10-CRC24}}$ \cite{elkelesh2018belief}}

\addplot[
    color=orange,
    mark=x,
    thick,
    mark size=2,
]
table {
2	325.97
2.5	145.935
3	99.6972
3.5	79.466
};
\addlegendentry{$\mathcal{T}_{\text{PBP-B10-CRC24}}$}

\addplot[
    color=blue,
    densely dashed,
    every mark/.append style={solid, fill=blue},
    mark=o,
    thick,
    mark size=2,
]
table {
2	388.09
2.5	151.1012
3	100.3984
3.5	79.4162
};
\addlegendentry{$\mathcal{T}_{\text{PBP-R32-CRC24}}$ \cite{elkelesh2018belief}}

\addplot[
    color=orange,
    mark=o,
    thick,
    mark size=2,
]
table {
2	307.706
2.5	143.1236
3	98.9794
3.5	79.242
};
\addlegendentry{$\mathcal{T}_{\text{PBP-B32-CRC24}}$}

\coordinate (spypoint1) at (axis cs:2,356);
\coordinate (magnifyglass1) at (axis cs:2.75,290);  
\end{axis}
\spy [blue] on (spypoint1) in node[fill=white] at (magnifyglass1);

\end{tikzpicture}
  \caption{Comparison between average decoding latency requirements of BP decoding on permuted factor graphs of $\mathcal{P}(1024,512)$ when the permutations are selected based on the proposed method in this paper and the methods in \cite{Kor09thesis, hussami2009performance} and \cite{elkelesh2018belief}. A $24$-bit CRC is used for early termination.}
  \label{latency}
\end{figure}

\section{Conclusion} \label{sec:conclude}
In this paper, we first introduced a mapping to represent a polar code factor graph permutation with an equivalent bit-index permutation. This scheme enables a single decoding scheduling that can be applied on different factor graph permutations, which is particularly interesting for hardware implementations. Second, we provided an empirical method to construct the good permutations of polar codes, in order to use only a small number of predetermined permutations to obtain a reasonable error-correction performance. We demonstrated that the error-correction performance of the proposed method applied to belief propagation (BP) decoding can obtain more than $0.25$~dB improvement at the frame error rate of $10^{-4}$, in comparison with BP decoding applied to the same number of randomly-selected permuted factor graphs for the 5G polar code of length $1024$ and rate $0.5$, concatenated with a $24$-bit cyclic redundancy check.


\end{document}